\documentclass[11pt]{amsart}

\usepackage{graphicx,verbatim}
\usepackage{amssymb}
\usepackage{epstopdf}
\usepackage[usenames]{color}
\usepackage{natbib}

\theoremstyle{plain}
\newtheorem{thm}{Theorem}[section]

\newtheorem{prop}[thm]{Proposition}
\newtheorem{cor}[thm]{Corollary}

\theoremstyle{definition}
\newtheorem{defn}[thm]{Definition}
\newtheorem{ex}[thm]{Example}

\theoremstyle{remark}
\newtheorem*{rmk}{Remark}

\newcommand{\PP}{\mathbb P}
\newcommand{\prob}{\mathbb P}
\newcommand{\EE}{\mathbb E}
\newcommand{\MRCA}{\text{MRCA}}

\newcommand{\tc}{\text{:}} 

\title{Species tree inference by the \textsl{STAR} method, and generalizations}

\date{\today}                    

\author{E.S. Allman}
\address{Department of Mathematics and Statistics, University of Alaska Fairbanks, PO Box 756660, Fairbanks, AK 99775 USA}
\email{e.allman@alaska.edu}

\author{J. H. Degnan}
\address{Department of Mathematics and Statistics,
              University of Canterbury
              Private Bag 4800,
              Christchurch, New Zealand}
\email{J.Degnan@math.canterbury.ac.nz}

\author{J. A. Rhodes}
\address{Department of Mathematics and Statistics, University of Alaska Fairbanks, PO Box 756660, Fairbanks, AK 99775 USA}
\email{j.rhodes@alaska.edu}

\begin{document}

\maketitle

\begin{abstract} The multispecies coalescent model describes the generation of gene trees
from a rooted metric species tree, and thus provides a framework for the inference of species trees from sampled gene trees.
We prove that the STAR method of Liu et al., and generalizations of it, are
statistically consistent methods of topological species tree inference under this model.  
We discuss the impact of gene tree sampling schemes for species tree inference
using generalized STAR methods, and reinterpret the original STAR as a consensus method based on
clades.
\end{abstract}


\section{Introduction}
While not always emphasized in practical analyses of biological
sequences, the gene trees which standard phylogenetic methods infer
may differ from the species tree relating the taxa from which the
samples are taken. A frequent reason for this incongruence between
gene trees and species trees is the population genetic effect of
incomplete lineage sorting. Although this source of incongruence has
been recognized for many years, recently methods have been proposed to
bring a statistical treatment of it into data analysis, through the
framework of the multispecies coalescent model
\citep{rannala2003,liu2007,kubatko2009,LiuSTAR,heled2010,Larget2010}.
However, the subject is still young, with both theoretical issues
concerning some of the methods not fully established, and practical
understanding of their appropriate application not yet developed.

\medskip

Here we focus on the STAR method, of 
\citet{LiuSTAR}, for inferring a rooted topological species tree from
rooted topological gene trees. It is a fast approach to species tree
construction that shows promise of good performance both in
simulations \citep{LiuSTAR} and in empirical studies \citep{Lee2012}.

Since STAR bases its inference on topological gene trees, it discards
all metric information one might have on them. While at first this may
seem wasteful, there are good reasons one might prefer such an
approach. In practice, the gene trees are likely to be themselves
inferred by standard phylogenetic methods, so their edge lengths are
often considered to be less reliably known than their topological
structure. Moreover, modeling the evolution of sequences on gene trees
that were produced under the coalescent model requires reconciling
time units on gene trees (reflecting total amounts of substitution),
with time units on the species tree (reflecting population size and
true time). This is usually resolved in part by assuming a molecular clock,
thus forcing gene trees to be ultrametric. Since variations in
base substitution rates along gene lineages can be large, and inferred
metric gene trees are often far from ultrametric, it is unclear what
impact this approach has on data analysis. The authors of the
Bayesian inference software BEST have, in fact, warned that violations
of the molecular clock could have `significant impacts'
\citep{BESTbookChap}.  By using only topological gene tree topologies,
STAR circumvents such problems.

STAR proceeds by first combinatorially encoding each topological gene
tree by a distance matrix. Averaging these over all gene trees
produces a distance matrix that is then used to build a species
tree. Finally, the metric information on the inferred species tree is
discarded, so only its topology is retained. Though \citet{LiuSTAR}
suggest UPGMA or Neighbor Joining (NJ) for this step, in fact any
distance method, including ones based on optimality criteria such as
Minimum Evolution, could be used. That an algorithmic tree
construction process such as Neighbor Joining is fast accounts only in part
for the speed of STAR. Its avoidance of calculating the
theoretical gene tree probabilities that would be needed in either
Maximum Likelihood or Bayesian analysis is also a major factor in its
speed, so use of other distance methods can still be expected to have
computational advantages.

\medskip

Although not investigated by \citet{LiuSTAR}, it is a rather
surprising fact that the expected distance matrix obtained by STAR
applied to the theoretical gene tree distribution under the
multispecies coalescent model will exactly fit a rooted tree with the
same topology as the species tree. That this is true is crucial to the
claim that STAR is a statistically consistent method of inference.  In
this work we prove this claim, and as a consequence establish the
statistical consistency of the STAR method when any of a large class
of distance methods are used for the final selection of the species
tree.

In the course of doing this, we also show that the method can be
generalized into a parameterized family of STAR methods that are all
consistent. The parameters specify a \emph{node numbering scheme}
which controls the way in which gene trees are encoded in distance
matrices. (Although the particular node numbering scheme of the original method is referred
to as a `ranking' by \citet{LiuSTAR}, we avoid that terminology as it
conflicts with the common use of the term `ranked tree' to indicate
that internal nodes are assigned distinct values, often to indicate
the temporal ordering of branching events in different parts of the
tree \citep{SempleSteel}.)  The original STAR method
arises from one particular choice of these parameters.

Our work also shows that under the coalescent the expected distance
tables obtained in all generalized STAR methods are ultrametric. This
is important for understanding the behavior of UPGMA, Neighbor
Joining, or any distance method when applied to the expected distance
table, a point that was not addressed by \citet{LiuSTAR}. It also
suggests that in the final selection of the species tree from the
average distances, it might be preferable to use a method that
enforces ultrametricity.

In closing, we show that the standard STAR method can be viewed as
using only the clade structure of the gene trees.  This observation
allows a reinterpretation of the STAR distance as one by which taxa
are viewed as close if they appear together in many clades across the
gene trees. The STAR method can thus be viewed as a type of greedy
consensus method using the clade structure of gene trees to infer the
species tree.  That it is a consensus method attuned in particular to the
    species tree/gene tree problem is the import of the fact that STAR
    is consistent under the coalescent model.

Although our emphasis in this paper is on theory, we also discuss some
possible implications for data analysis, focusing on various sampling
schemes and how generalizations of the STAR method might possibly be
used to gain a measure of confidence in species tree estimates.  
Additional work is needed to further develop the suggestions here.

\medskip

\section{The standard STAR Node Numbering Scheme and Variants}

The encoding of an $n$-taxon binary rooted gene tree topology by a
distance matrix in STAR is accomplished through first associating
numbers to each node of the tree. The \emph{standard node numbering
  scheme} used by STAR is as follows: Assign $n$ to the root of the
tree.  Then if an internal node has been assigned a number $m$, assign
to each of its children that are internal nodes the number $m-1$.
These numbers are then interpreted as distances from the leaves in an
ultrametric tree, so that the distance between two leaves is twice the
number assigned to their most recent common ancestor.
In Figure \ref{fig:numbering}, for example, the
  distance between genes $A$ and $C$ is $6$ on the left tree,
  and $10$ on the right tree.

\begin{figure}[t]
\begin{picture}(45.0,170)(0,0)
\put(-150,0){
\includegraphics[width=.48\textwidth]{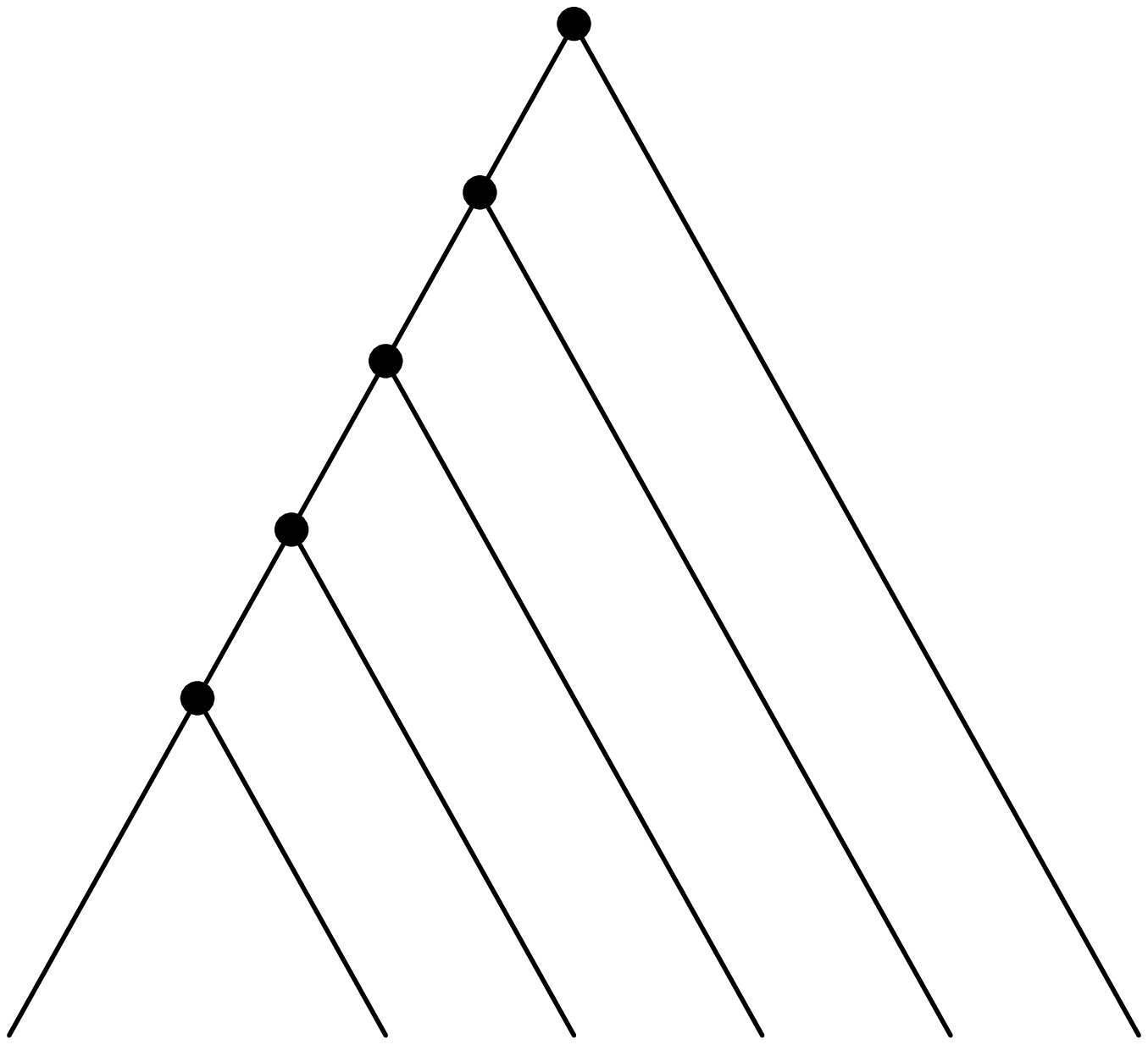}}
\put(10,0){\includegraphics[width=.48\textwidth]{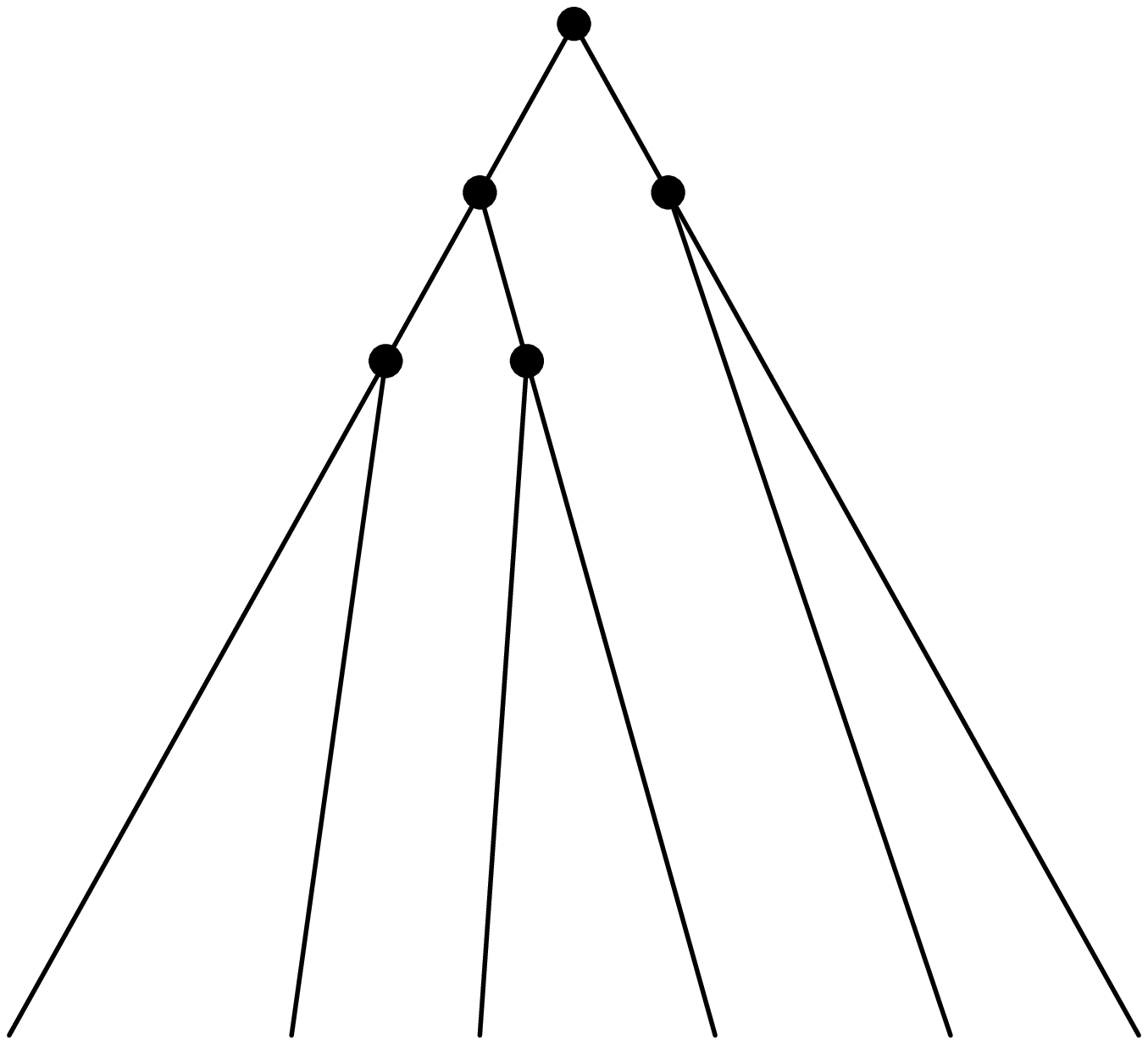}}
 \put(-69,155){\fontsize{11.23}{14.07}\selectfont   \makebox(600.0, 0.0)[l]{$a_0$ \strut}}
  \put(-80,135){\fontsize{11.23}{14.07}\selectfont   \makebox(600.0, 0.0)[l]{$a_1$ \strut}}
   \put(-91,115){\fontsize{11.23}{14.07}\selectfont   \makebox(600.0, 0.0)[l]{$a_2$ \strut}}
    \put(-103,95){\fontsize{11.23}{14.07}\selectfont   \makebox(600.0, 0.0)[l]{$a_3$ \strut}}
     \put(-115,75){\fontsize{11.23}{14.07}\selectfont   \makebox(600.0, 0.0)[l]{$a_4$ \strut}}
      \put(88,155){\fontsize{11.23}{14.07}\selectfont   \makebox(600.0, 0.0)[l]{$a_0$ \strut}}
  \put(77,135){\fontsize{11.23}{14.07}\selectfont   \makebox(600.0, 0.0)[l]{$a_1$ \strut}}
   \put(66,115){\fontsize{11.23}{14.07}\selectfont   \makebox(600.0, 0.0)[l]{$a_2$ \strut}}
    \put(118,135){\fontsize{11.23}{14.07}\selectfont   \makebox(600.0, 0.0)[l]{$a_1$ \strut}}
     \put(99,115){\fontsize{11.23}{14.07}\selectfont   \makebox(600.0, 0.0)[l]{$a_2$ \strut}}
     \put(-127,20){\fontsize{11.23}{14.07}\selectfont   \makebox(600.0, 0.0)[l]{$A$ \strut}}    
        \put(-82,20){\fontsize{11.23}{14.07}\selectfont   \makebox(600.0, 0.0)[l]{$B$ \strut}}    
             \put(-59,20){\fontsize{11.23}{14.07}\selectfont   \makebox(600.0, 0.0)[l]{$C$ \strut}}    
                \put(-37,20){\fontsize{11.23}{14.07}\selectfont   \makebox(600.0, 0.0)[l]{$D$ \strut}}
                  \put(-15,20){\fontsize{11.23}{14.07}\selectfont   \makebox(600.0, 0.0)[l]{$E$ \strut}} 
                    \put(7,20){\fontsize{11.23}{14.07}\selectfont   \makebox(600.0, 0.0)[l]{$F$ \strut}}
                     \put(30,20){\fontsize{11.23}{14.07}\selectfont   \makebox(600.0, 0.0)[l]{$A$ \strut}}    
        \put(62,20){\fontsize{11.23}{14.07}\selectfont   \makebox(600.0, 0.0)[l]{$B$ \strut}}    
             \put(84,20){\fontsize{11.23}{14.07}\selectfont   \makebox(600.0, 0.0)[l]{$C$ \strut}}    
                \put(113,20){\fontsize{11.23}{14.07}\selectfont   \makebox(600.0, 0.0)[l]{$D$ \strut}}
                  \put(140,20){\fontsize{11.23}{14.07}\selectfont   \makebox(600.0, 0.0)[l]{$E$ \strut}} 
                    \put(165,20){\fontsize{11.23}{14.07}\selectfont   \makebox(600.0, 0.0)[l]{$F$ \strut}}
\end{picture}
\caption{ Two 6-taxon gene trees with a STAR node-numbering.  Both
  trees are drawn ultrametricly according to the standard STAR numbering
  with $(a_0,a_1,a_2,a_3,a_4) = (6,5,4,3,2)$.  Note that only the
  caterpillar tree shape (left) has distinct numberings for each node.
  For a generalized node numbering, the $a_i$ can be any non-negative
  decreasing sequence with at least two distinct terms.  Node numbers
  are interpreted as remetrizing the gene tree, by specifying
  distances of the nodes from their descendant leaves.}
\label{fig:numbering}
\end{figure}

An alternative viewpoint is that the standard STAR metrization of a gene tree is 
obtained from the rooted gene tree topology by assigning all internal edges length 1, 
and then choosing lengths of pendant edges so that all leaves are distance $n$ from the root.

It is immediately clear that some trivial variations on this numbering scheme will have no effect 
on the output of STAR. For instance, suppose instead of assigning the root the number $n$, one 
assigned to it any number $\ell \ge n-2$, but otherwise followed the standard scheme. Then all 
internal branches would still have length 1, and pendant ones would have non-negative length, 
regardless of the topology of the gene tree. Indeed pendant edge lengths would only change 
by the addition of $\ell-n$. The net effect on the pairwise distance matrix for each gene tree 
is to add  $2(\ell-n)$ to each off-diagonal entry. 
Thus the empirical average of these matrices over a sample of gene trees, and the theoretical 
expected distance matrix change in the same way. (Note that these distance matrices will 
still have non-negative off-diagonal entries.) It is easy to see that this will have no effect on the 
topology of the species tree output by methods such as UPGMA and NJ. Similarly, assuming 
one chooses a sufficiently large number to assign to the root, one could number descendant 
nodes by any constant decrement from their parents, with no effect on the output of the 
topological species tree by UPGMA or NJ, since this would merely change the average 
distance matrix by a rescaling and addition of a constant.
  
Thus assuming the output of  standard STAR from the  theoretical gene tree distribution 
does agree topologically with the species tree, there are certainly some variations on the 
numbering scheme that will have the same property. Less obvious variations of node 
numberings that can also be used for species tree inference in a generalized STAR 
method are the focus of the next section.

\medskip

By way of contrast, one might consider the following seemingly natural node numbering: 
Assign to each node the count of its leaf descendants.
That this
`descendant-count' numbering scheme
behaves differently from the STAR scheme is seen by considering the 
$5$-taxon caterpillar species tree $$\sigma=((((a,b)\tc x,c)\tc y,d)\tc z,e),$$ 
where $x,y,z$ are in coalescent units. Taking $x=z=\infty$ and $y=0$, the only 
gene trees with non-zero probability under the coalescent model are
\begin{align*}
T_1=((((a,b),c),d),e),&\ \ \PP_\sigma(T_1)= 1/3, \\
T_2=((((a,b),d),c),e),&\ \ \PP_\sigma(T_2)= 1/3,  \\
T_3=(((a,b),(c,d)),e),&\ \ \PP_\sigma(T_3)=  1/3. 
\end{align*}
Using these probabilities as weights in averaging the distance matrices 
$$D_{T_1}=\begin{pmatrix}0&4&6&8&10\\4&0&6&8&10\\6&6&0&8&10\\8&8&8&0&10\\10&10&10&10&0\end{pmatrix},  \hskip .5cm
D_{T_2}=\begin{pmatrix}0&4&8&6&10\\4&0&8&6&10\\8&8&0&8&10\\6&6&8&0&10\\10&10&10&10&0\end{pmatrix},$$
$$D_{T_3}=\begin{pmatrix}0&4&8&8&10\\4&0&8&8&10\\8&8&0&4&10\\8&8&4&0&10\\10&10&10&10&0\end{pmatrix},$$
we see
$$\EE_\sigma(D_T)= \begin{pmatrix}0&4&22/3&22/3&10\\4&0&22/3&22/3&10\\22/3&22/3&0&20/3&10\\22/3&22/3&20/3&0&10\\10&10&10&10&0\end{pmatrix}.$$
This expected distance matrix exactly fits the rooted tree
$$(((a\tc 2,b\tc 2)\tc 5/3, (c\tc 10/3, d\tc 10/3)\tc 1/3)\tc 4/3,  e\tc 5)$$
which is topologically different from the species tree, even if both are treated as unrooted. 
Similar examples with positive and finite edge lengths can easily be constructed,
by  taking $x,z$ to be large while $y$ is small, since the expected distance matrix depends continuously on the edge lengths.

If, instead, the standard STAR numbering were used in the above example, 
$D_{T_1}$ and $D_{T_2}$ would be unchanged, but in $D_{T_3}$  all entries 
of $4$ would be replaced by $6$. How this effects $\EE_\sigma(D_T)$, and 
thus changes the tree it fits, we leave as an easy exercise for the reader. 

\section{Consistent node numbering schemes }\label{sec:consistnode}

We next show that any node numbering scheme with several simple
properties, when used in STAR with UPGMA, NJ, or any well-behaved
distance method for the selection of the species tree, leads to
consistent inference. By `well-behaved' we simply mean one that when
applied to any distance table in some neighborhood of one exactly
fitting $T$ will return the topology of $T$ .  

The properties we require for a node numbering scheme are:
\begin{enumerate}

\item \label{item:anc} The number assigned to any internal node on a
  gene tree depends only on the count of nodes between it and the root
  in the topological gene tree. Thus the node numbers on all gene
  trees come from a common \emph{node numbering sequence}
  $a_0,a_1,\dots, a_{n-2}$, with $a_0$ assigned to all roots, $a_1$ to
  their children, and more generally $a_i$ to all internal nodes at
  depth $i$ from the root.
 
\item \label{item:dec} A non-negative number is assigned to every
  node, with the number assigned to a child always less than or equal
  to that assigned to its parent, with at least one instance of 
strict inequality.  That is, $a_0\ge a_1\ge \dots \ge a_{n-2}\ge a_{n-1}
 = 0$ with some $a_\ell>a_{\ell+1}$ for some $\ell$.
 
\end{enumerate}

Equivalently, one could define an \emph{edge length sequence}
$b_1,b_2,\dots, b_{n-1}$, where $b_i\ge 0$, and at least one $b_i>0$,
and then assign length $b_i$ to all {internal} edges at depth $i$ from
the root. If a {pendant} edge is at depth $i$, it is assigned the
length $\sum_{j\ge i} b_j$.  This is related to a node numbering
sequence by defining $b_i=a_{i-1}-a_i$.

Both of these properties are natural. Property \eqref{item:dec}
ensures all internal branch lengths are non-negative, with at least
one positive, while property \eqref{item:anc} ensures a connection to
the number of coalescent events in a gene's lineage that occur above
(\emph{i.e.}, temporally before) the specified node.

Note that the standard STAR scheme is simply the special choice of the node
numbering sequence
$n,n-1,n-2,\dots, 2$, or, equivalently, of
the edge length sequence $1,1,1,\dots, 2$.
On the other hand, the
descendant-count
scheme mentioned in the previous section
fails to satisfy property \eqref{item:anc}. Though it is
tied to counting coalescent events, that scheme counts events
\emph{below} the node.

\medskip

We consider first a generalized STAR method where for each gene
exactly one individual is sampled per taxon. 
Extensions to sampling multiple individuals will be
discussed later.

\medskip

\begin{defn} A $n\times n$ numerical matrix, with rows and columns
  indexed by an $n$-element set of taxa $\mathcal X$, \emph{weakly fits} a
  topological tree $\psi$ on $\mathcal X$ if there is some assignment of
  non-negative lengths to the edges of $\psi$ such that the resulting
  pairwise tree distances between leaves give the entries of the
  matrix. If the edge lengths are positive, then we say the matrix
  \emph{strongly fits} $\psi$. If the assignment of edge lengths
  yields an ultrametric tree, we say the matrix fits $\psi$
  \emph{ultrametricly}.
\end{defn}

For a complete presentation of the multispecies coalescent model,
which describes the generation of gene trees from a metric species
tree, we refer to earlier works, for instance \citep{adr2011a}. The
primary feature of the coalescent model that we use is the 
\emph{exchangeability of lineages}: for a collection of lineages 
present in a population at a particular time, 
the probability of any pattern of coalescence (moving
backwards in time) of these lineages is the same for all permutations
of the lineages.

We denote a species tree on a set of taxa $\mathcal X$ by
$\sigma=(\psi,\lambda)$, where $\psi$ is a rooted topological tree and
$\lambda$ is a vector of edge lengths measured in coalescent units. We use lower case letters
$i,j,k$ to denote taxa, and the corrresponding upper case letters
$I,J,K$ to denote gene samples from those taxa.
For any $\mathcal Y\subseteq \mathcal X$, by the most recent common ancestor
of $\mathcal Y$, MRCA$(\mathcal Y)$, we mean the vertex on $\psi$ that is
ancestral to all elements of $\mathcal Y$, and a descendant of any
other vertex ancestral to all of $\mathcal Y$.  Thus $\mathcal Y$ is
always a subset of the set of descendants of MRCA$(\mathcal Y)$.
An MRCA on a gene tree is defined similarly.

\smallskip

Our main result is the following.

\begin{thm}\label{thm:dist}
  Under the multispecies coalescent model on a metric species tree
  $\sigma=(\psi,\lambda)$, let $\mathbb P_\sigma(T)$ denote the
  probability of a rooted topological gene tree $T$. For a node
  numbering sequence satisfying properties \eqref{item:anc} and
  \eqref{item:dec}, let $D_T$ denote the pairwise distance matrix
  obtained from applying the sequence to $T$. Then for all
  $\lambda_i\in(0,\infty)$ the expected value of this matrix,
$$\EE_\sigma (D_T)= \sum_T \PP_\sigma(T) D_T$$
is a pairwise distance matrix that strongly fits $\psi$ ultrametricly. 
\end{thm}

Of course the pairwise distance matrix $\EE_\sigma (D_T)$ of this
theorem does not generally match the pairwise distances on the species
tree $\sigma$.

\begin{proof}[Proof of Theorem \ref{thm:dist}] For a binary
  topological species tree $\psi$ with positive edge lengths
  $\lambda_i$, it is enough to show that, for any three taxa $i,j,k$,
  if $\psi$ displays the rooted triple $((i,j),k)$, then with
  $D=\EE_\sigma (D_T)$,
\begin{equation*}D(i,j)<D(i,k)=D(j,k).
\end{equation*}
This is simply the $3$-point condition for a distance to determine an ultrametric tree \citep{SempleSteel}.

\smallskip

So suppose $\psi$ displays $((i,j),k)$, and let $w=\text{MRCA}(\{i,j,k\})$ on $\sigma$.  Using shorthand
such as $C(I,J)<w$ to mean the coalescence of the (gene) lineages of
$I$ and $J$ occurs below node $w$ in the species tree, and
$\PP_\sigma(T, C(I,J)<w)$ to mean the joint probability of the gene
tree $T$ and the event $C(I,J)<w$, we see
\begin{align}
D(i,j)&=\sum_T \PP_\sigma(T) D_T(i,j)\notag \\
&= \sum_T \PP_\sigma(T, C(I,J)<w) D_T(i,j)
+ \sum_T\PP_\sigma(T,C(I,J)>w) D_T(i,j). \label{eq:eq1}
\end{align}
But we claim
\begin{equation}
\sum_T \PP_\sigma(T, C(I,J)<w) D_T(i,j)<\sum_T \PP_\sigma(T, C(I,J)<w) D_T(i,k).\label{eq:eq2}
\end{equation}
To see this, first note that if the lineages of $I,J$ coalesce below
$w$ and a gene tree $T$ is observed, then $T$ must display the rooted
triple $((I,J),K)$.  But for such gene trees, $D_T(i,j)\le
D_T(i,k)$. Thus a weak version of inequality \eqref{eq:eq2} holds
termwise.

Note now that under the coalescent model on a species tree 
with finite edge lengths, if $C(I,J)<w$, then every labeled gene tree 
$T$ which displays $((I,J),K)$ is realizable.  In particular, since $a_\ell>a_{\ell+1}$,
and there is a gene tree displaying $((I,J),K)$ for which $\MRCA(I,J)$
is at depth $\ell+1$ or greater, and the $\MRCA(I,K)$ is at depth
$\ell$ or less, one has $D_T(i,j)< D_T(i,k)$ for some gene tree $T$
with $\PP_{\sigma}(T,C(I,J)<w)>0$. Thus at least one pair of
corresponding terms in inequality \eqref{eq:eq2} satisfies a strict
inequality, and \eqref{eq:eq2} is established.

We claim also that
\begin{equation}\sum_T\PP_\sigma(T, C(I,J)>w) D_T(i,j)=\sum_T\PP_\sigma(T, C(I,J)>w ) D_T(i,k).\label{eq:eq3}
\end{equation}
Indeed, 
up to multiplication by a positive constant,
these sums respectively give the conditional expectation of $D(i,j)$
and $D(i,k)$, given that the lineages of $I,J,K$ are distinct at
$w$.
But by property \eqref{item:anc} of the numbering scheme, and
the exchangeability of lineages under the coalescent, these conditional
expectations are the same.

From equations \eqref{eq:eq1}, \eqref{eq:eq2}, and \eqref{eq:eq3}
the inequality $D(i,j)<D(i,k)$ now follows.

To see that $D(i,k)=D(j,k)$, observe
\begin{align*}
D(i,k)&=\sum_T \PP_\sigma(T) D_T(i,k)\\
&= \sum_T \PP_\sigma(T, C(I,J)<w) D_T(i,k)+ \sum_T\PP_\sigma(T, C(I,J)>w) D_T(i,k). 
\end{align*}
However
$$\PP_\sigma(T, C(I,J)<w) D_T(i,k) =\PP_\sigma(T, C(I,J)<w) D_T(j,k),$$ 
since if the probability appearing in this equation is
non-zero then the MRCA of $I,K$ on $T$ is the same as the MRCA of
$J,K$.  
Moreover,
$$
 \sum_T\PP_\sigma(T, C(I,J)>w) D_T(i,k)=\sum_T\PP_\sigma(T, C(I,J)>w ) D_T(j,k)
$$
by an exchangability argument similar to that used in establishing
inequality \eqref{eq:eq3}. Thus $D(i,k)=D(j,k)$,
and the $3$-point condition is established.

\smallskip

For non-binary $\psi$, one must also check that if $\psi$ displays the
unresolved $3$-taxon subtree $(i,j,k)$, then
 $$D(i,j)=D(i,k)=D(j,k).$$
 This is done using exchangability of lineages, as in the 
 justification for equation \eqref{eq:eq3}.
\end{proof}

\begin{cor} \label{cor:statcon} Consider the generalized STAR method
  using any node numbering sequence satisfying properties
  \eqref{item:anc} and \eqref{item:dec}, and any method of tree
  selection that from a distance table strongly fitting a binary tree
  returns that tree, and whose output is continuous at such tables.
  Then under the multispecies coalescent model on $\sigma=(\psi,
  \lambda)$, with $\psi$ binary and $\lambda_i>0$, the method is
  statistically consistent for inferring the species tree topology
  $\psi$.\end{cor}

\begin{proof} Let $N$ be the number of sampled gene trees, and
  $\epsilon>0$.  By the law of large numbers, as $N\to \infty$, the
  probability that the average distance matrix differs from $\EE_\sigma(D_T)$
  by more than $\epsilon$ in any entry approaches 0. Thus by
  continuity of the distance method of tree selection, as $N\to
  \infty$ the probability that the inferred tree will
  have the same topology as that from $\EE_\sigma(D_T)$, $\psi$,
   approaches $1$. 
 \end{proof}

\citet{Atteson97} showed that the Neighbor Joining tree construction algorithm
satisfies the continuity
hypothesis of this corollary.  A proof that UPGMA is continuous  is straightforward.

\begin{ex}\label{Ex:extreme} As an extreme example of a node numbering
  scheme leading to consistent inference, for $n$ species let $a_i = 2$
  for $0 \le i < n-2$ and $a_{n-2} = 1$.  Interpreting the $a_i$
  values as distances from the leaves, this scheme converts every non-caterpillar gene
  tree into a completely unresolved tree, and every caterpillar gene
  tree into a gene tree with one non-trivial clade of two
  leaves.  In effect, the resulting generalized STAR method discards 
  all information gathered from non-caterpillar gene trees yet, by Theorem
  \ref{thm:dist}, the species tree topology can still be recovered.  Thus it is possible, in principle, to 
  reconstruct the species tree topology using only caterpillar gene trees 
  as input, even if the species tree is not a caterpillar.  Of course,
  we do \emph{not} recommend this for practical inference.
\end{ex}

\medskip

It is natural to focus especially on node numbering sequences that are
strictly decreasing. Indeed, if all branch lengths
 on $\sigma=(\psi,\lambda)$ are very long, then with high probability a
finite sample of topological gene trees will include only those
matching $\psi$.  If a node numbering is not strictly decreasing,
however, then an estimate of $\EE_\sigma(D_T)$ from the sample would give
$D_\psi$, which may only weakly fit $\psi$, leading to a
poorly-resolved inferred species tree under STAR.

The next theorem shows that strictly decreasing node numbering sequences lead to
better behavior of generalized STAR methods, in the sense that longer internal branches 
on the species tree do not make it more difficult to infer a fully-resolved
species tree from a fixed sample size of gene trees. Since longer
edges in a species tree increase the probability of the gene trees
showing the associated split, this is intuitively desirable, and would certainly
be a useful characteristic for data analysis.

\smallskip

Recall now that the probability of any gene tree under the coalescent is a polynomial
in the transformed species tree branch lengths, $\exp(-\lambda_i$), as
is explained, for instance, by \citet{adr2011a}. Thus for any
collection $\{D_T\}$ of matrices associated to gene trees (whether or
not it arises from a node numbering sequence), the expected value
$\EE_\sigma(D_T)$ also has entries that are polynomial in the
$\exp(-\lambda_i)$. As a result, the expected value also makes sense for a
branch length of $\infty$, by setting the transformed branch length to
$0$.  (Treating a branch length of $\infty$ in this way is
equivalent to taking a limit as the branch length goes to $\infty$.)
While an infinite branch length of course has no real biological
meaning, we will allow branch lengths in
$(0,\infty]=(0,\infty)\cup\{\infty\}$, both to simplify the
presentation of arguments, and to easily describe behavior as branches
grow long.

Note also that the strong fitting of $\psi$ by $\EE_\sigma(D_T)$ for finite
$\lambda$ is expressible by  equalities and strict inequalities
in the matrix entries, arising from the 3-point conditions. By continuity then, if some $\lambda_i\to
\infty$ the same equalities and non-strict versions of the
inequalities must hold. That is, if branch lengths are in $(0,\infty]$,
then $\EE_\sigma(D_T)$ will certainly weakly fit $\psi$.  By requiring
that the node numbering is strictly decreasing, however, this can be improved.

\begin{thm}
  Suppose a node numbering sequence is strictly decreasing. Then for
  any $\psi$ and all $\lambda_i\in(0,\infty]$, the expected distance
  matrix $\EE_\sigma(D_T)$ strongly fits $\psi$, ultrametricly.

  Moreover, if the node numbering sequence is not strictly decreasing, then
  there is a binary $\psi$ and choices of $\lambda_i\in(0,\infty]$ for which
  $\EE_\sigma(D_T)$ only weakly fits $\psi$.

\end{thm}

\begin{proof}  
  The proof of the first claim follows the argument of Theorem
  \ref{thm:dist}. The primary modification is in the justification of
  inequality \eqref{eq:eq2}, when some of the internal edge lengths are
  $\lambda_i=\infty$.  But since a strictly decreasing node numbering
  implies $D_T(i,j)<D_T(i,k)$ for all gene trees $T$ displaying
  $((I,J),K)$, the inequality is immediately clear.

  For the second claim, suppose $a_\ell=a_{\ell+1}$ and pick any
  binary species tree $\psi$ which has an internal edge between nodes
  $u$ of depth $\ell$ and $v$ of depth $\ell+1$. Assign edge lengths
  of $\infty$ to all edges incident to the ancestors of $u$, to the
  edge descending from $u$ that is not incident to $v$, and to the two
  edges descending from $v$. Pick taxa $i,j$ that are descendants of
  $v$ through the two different edges, and taxon $k$ that is a
  descendant of $u$ but not of $v$.    Then regardless of the other edge
  lengths, the only gene trees realizable under the coalescent will be
  those with one of $\MRCA(I,J)$, $\MRCA(I,K)$, and $\MRCA(J,K)$ at
  depth $\ell+1$, say $\MRCA(I,J)$, and thus $\MRCA(I,K)=\MRCA(J,K)$
  located at depth $\ell$ as the parent of $\MRCA(I,J)$.
  Since $a_\ell=a_{\ell+1}$, we have
  $\EE_\sigma(D_T)(i,j)=\EE_\sigma(D_T)(i,k)=\EE_\sigma(D_T)(j,k)=2a_\ell$
  are all equal, so $\EE_\sigma(D_T)$ cannot strongly fit $\psi$.
\end{proof}

\begin{rmk}
 While we have established that properties \eqref{item:anc}
  and \eqref{item:dec} are sufficient for statistical consistency of
  STAR, we have not shown they are necessary.  In fact, addressing
  this question for the full STAR process which includes the selection
  of the species tree by some algorithm or optimality criterion would
  be difficult, since even if an expected distance matrix does not
  exactly fit a tree, the fitting processes in the distance method may
  overcome the misfit.

Avoiding the impact of the final tree selection step, one can investigate the existence of alternative schemes in which 
  distance matrices might be assigned to gene trees in such a way that for any species tree
  the expected distance matrix under the coalescent fits the
  topological species tree. Preliminary investigations indicate that
  at least for small trees such assignments exist that do not arise
  from generalized STAR numbering schemes.
\end{rmk}

The theorems above were stated and proved 
under a gene tree sampling scheme in which one individual is sampled per taxon.
We next consider extensions to multi-individual sampling.

In the case of a sampling scheme in which $k_i$ individuals are
sampled from taxon $i$ for every gene, a simple device for extending STAR
is the \emph{extended species trees}:
At leaf $i$
on the original species tree attach $k_i$ descendant edges in a
multifurcation, leading to new taxa $i_j$, $1\le j\le k_i$. These new
pendant edges need not be assigned lengths, but  
all edges arising from the original species tree retain their lengths.

The
coalescent model on the original species tree with multiple individuals
sampled per taxon then produces exactly the same distribution
of gene trees as the coalescent on the extended species tree with one sample
per taxon. Thus Theorem \ref{thm:dist} applies.

Since this extended species tree is not binary if any $k_i > 2$, Corollary
\ref{cor:statcon} does not apply directly. However, the procedure
proposed by \citet{LiuSTAR} for multi-individual sampling is to first
average over these individuals so that a distance matrix is obtained
relating the original taxa. We sketch an argument that this leads to
consistent inference. Since the average distance table for the gene trees will approximate the expected one, which by
Theorem \ref{thm:dist} strongly fits the extended species tree
ultrametrically, so will the result of further averaging this
empirical distance table over its image under all permutations of
individuals within each taxon. In this averaged table, distances
between individuals in different taxa depend only on the two taxa, and
not the individuals. Deleting all but one individual per taxon from
this table will then yield a table that in expectation strongly fits
the original species tree, and whose entries are those calculated by the
procedure of \citet{LiuSTAR}.  Thus as long as the original species
tree is binary, one still has a consistent inference scheme.
For an example of an empirical study that uses STAR with a
variable number of individuals per species, see \citet{Lee2012}.

An even more general sampling scheme might vary the number of individuals
sampled by both taxon and gene.  For instance if taxon $i$ is sampled
$k_{ij}$ times for locus $j$, we could define a sampling scheme for
locus $j$ as $\mathbf{k}_j = (k_{1j}, \dots, k_{nj})$, where $n$ is
the number of species.  Letting $\mathcal K = \{\mathbf k\}$ be the finite set of all
possible schemes used in a particular study design, if each locus has
a sampling scheme independently chosen from some distribution on
$\mathcal K$, then as the number of loci approaches infinity, the
average empirical distance table for each sampling scheme
approaches the expected distance table for that sampling scheme by the
Law of Large Numbers.  Thus, the expected distance table over $\mathcal
K$ is a weighted average over the expected distance tables for each
$\mathbf{k} \in \mathcal K$, and hence also ultrametrically fits the
species tree by the $3$-point condition. Thus the multiple-allele
version of STAR is consistent method of inference even with different
numbers of alleles sampled at different loci.
 
Finally, one could generalize further to allow different STAR node
numberings for different sampling schemes. More precisely, if $\mathbf
k_j \in \mathcal K$ specifies both a sampling scheme and numbering
scheme for locus $j$, chosen according to some distribution on
$\mathcal K$, then the STAR method will still be statistically
consistent. Of course with such an approach an increased number of loci are likely to be needed for
observed average distances to be close to their expected values.

 \section{Finite samples of loci}
 
 The statistical consistency of STAR methods, as established in the
 preceding section, is a statement about asymptotic behavior as the
 number of sampled loci approaches infinity. In this section we
 collect some observations on possible pitfalls of naively applying
 STAR when the number of sampled loci is small.
  
 \medskip
 
 We begin with issues concerning multiple samples of individuals
 within taxa, where the number of individuals varies with the locus.
 
 First, imagine a rather extreme sampling scheme in which $N$ loci are
 sampled, but that at locus $j$, taxon $i$ is sampled from $i+j$
 individuals.  Thus each sampling scheme is used for exactly one
 locus.  The statistical consistency of STAR in this situation is not
 helpful to us, since without multiple uses of every sampling scheme
 we do not expect to see asymptotic behavior.  The Law of Large
 Numbers does not apply, and the observed average distance table might
 be quite different from its expectation.

 If in a real study the number of sampling schemes used at different
 loci is very large in comparison to the number of loci,
 an appeal to the Law of Large Numbers may be similarly questionable.  For
 example, if each of $n$ species has between 1 and 3 individuals
 sampled at each locus, then there are potentially $3^n$ sampling
 schemes. We would need to sample enough loci to ensure that each of
 these exponentially many sampling schemes is used many times for the
 asymptotic claim of consistency to give us much confidence in an
 analysis.
 
 If, on the other hand, a study is designed for 3 individuals sampled
 per species and sequenced at all loci, a small amount of missing data
 might result in only a few of the $3^n$ sampling schemes being used,
 with each used repeatedly. Especially if the data were missing at
 random, it might be reasonable to use the multiple-allele versions of
 STAR.

 As another example of potentially problematic sampling, imagine
 sampling a large number of loci with only one individual per taxon,
 and a single locus with one individual for most taxa, but many
 individuals for one taxon. (Such a data set is not farfetched,
 especially if it were built from data originally collected for other
 studies.)  Then using the standard STAR numbering, the tree for this
 last locus will have a much larger number assigned to its root. If
 there is substantial incomplete lineage sorting between the
 single-sample taxa and the multi-sample one, then on this tree
 internode distances between coalescent events involving only
 single-sample taxa will be larger than for the other trees. Thus it
 appears that it might be desirable to downweight this tree in the
 distance averaging, to avoid overemphasizing the evidence for
 relationships between the single-sample taxa that it offers. On the
 other hand, if there is little incomplete lineage sorting between the
 single-sample taxa and the multi-sample one, then downweighting will
 result in an undesirable deemphasis of the evidence of relationships
 between the single-sampled taxa given by this locus.  It is simply
 not clear how the anomalously-sampled locus could be used without
 potentially biasing the analysis.

\medskip

We note too that for finite sample sizes, the standard STAR numbering
scheme and its variations may have biases for tree shape.  To illustrate this,
consider the case of sampling a single individual per taxon for four taxa.

Now suppose UPGMA is applied to the
average distance table using the standard STAR numbering with
$(a_0,a_1,a_2) = (4,3,2)$, with ties being resolved randomly.  If
there are exactly two input gene trees, then STAR-UPGMA will return a
symmetric tree in exactly the following cases: 

\begin{enumerate}

\item[(i)] both input trees are symmetric; 

\item[(ii)] one input tree is symmetric and one is asymmetric, and
they share a \emph{cherry} (a $2$-taxon clade), in which case a
symmetric tree is returned with probability 1/2; and 

\item[(iii)] both input
trees are asymmetric, and their cherries do not have overlapping taxa,
in which case a symmetric tree is returned with probability 2/3.  

\end{enumerate}

In case (iii), for instance, if the input trees are (((AB)C)D) and
(((CD)B)A), then the expected distance table averaged over both loci
is

$$\EE(D_{T})= \begin{pmatrix}0&6&7&8\\6&0&6&7\\7&6&0&6\\8&7&6&0\end{pmatrix}, $$
which results in an asymmetric tree if B and C are clustered first
(which occurs with probability 1/3), and otherwise results in a
symmetric tree.

Assuming that gene trees evolve on an unresolved species tree, 
unlabeled gene
trees have a 1/3 probability of being symmetric and 2/3 probability of
being asymmetric under the coalescent.  The probability that
STAR-UPGMA returns a symmetric tree in this situation is therefore
$$\frac{1}{3}\cdot\frac{1}{3} + 2 \cdot \frac{1}{3}\cdot\frac{2}{3}\cdot\frac{1}{3}\cdot\frac{1}{2} + \frac{2}{3}\cdot\frac{2}{3}\cdot\frac{1}{6}\cdot\frac{2}{3} = \frac{19}{81} < \frac 13,$$
indicating that there is some bias against symmetric trees in this
case.  (To understand the second summand in the expression on the left, 
for example, notice that the two trees are symmetric and asymmetric
with probability $2 \cdot \big(\frac{1}{3}\big) \big( \frac{2}{3}\big)$, that
these two trees share a cherry with probability $\frac{1}{3}$, and that
UPGMA will return the symmetric tree with probability $\frac{1}{2}$.)

While these computations show bias is present, we note 
that this analysis has discarded the edge
lengths on the UPGMA tree, even though they carry information about
when a tree is close to a different topological structure.

Now it is possible that numbering schemes can be chosen to reduce bias
when only finitely many loci are sampled.  For example, for the
numbering scheme $(a_0,a_1,a_2) = (4,3,1)$, case (iii) always results
in STAR returning a balanced tree, and the total probability of
returning a symmetric tree is $21/81$, suggesting less bias than the
standard STAR numbering, 
at least for this example.

Whether topological bias of STAR for small samples can be reduced for
all species trees simultaneously, or at least on average under some
model of species tree generation, is an interesting topic for future investigation.

\section{STAR and clades} 

The standard node numbering scheme, in which all internal edges on
gene trees are given length 1, produces distances on remetrized gene
trees that has an interpretation regarding clades on
gene trees. In this section we develop this connection.

As a consequence, we recover a previously known result that the
topological rooted species tree can be identified from the collection
of gene clade probabilities \citep{adr2011b}.  Additionally, with this
viewpoint, we obtain an efficient algorithm for the identification.

This connection between STAR and clades makes clear that whatever
information STAR uses for inference is found in the clade
probabilities, and not in the more detailed distribution of gene
trees.

\medskip

Consider a gene tree $T$ on a collection of taxa $\mathcal X$.
By a clade on $T$ we mean a subset $\mathcal C\subseteq \mathcal X$
for which the set of descendants of MRCA$(\mathcal C)$ is precisely
$\mathcal C$. There is thus a bijection between nodes $v$ of $T$ and
clades $\mathcal C$, by which $v=\text{MRCA}(\mathcal C)$. The trivial
clades are $\mathcal X$, with MRCA at the root of the tree, and
singleton subsets of $\mathcal X$, with MRCAs at the leaves. The
trivial clades occur on all gene trees on $\mathcal X$.

Given any two leaves $a,b$ of $T$, let $H_{a,b}$ denote the set of
non-trivial clades of $T$ which contain both $a$ and $b$, and let
$c_T(a,b)=|H_{a,b}|$.  Elements of $H_{a,b}$ correspond to the
non-root vertices of $T$ lying on the path from the root to the
MRCA$(\{a,b\})$. Thus $c_T(a,b)$ is simply the distance from the root
to the MRCA$(\{a,b\})$ in $T$ remetrized by the standard STAR
numbering. Since this tree is ultrametric, with all leaves $n=|X|$
from the root, one finds
\begin{equation}
c_T(a,b)=n-\frac {D_T(a,b)}2.
\label{eq:cd}
\end{equation}
This shows that, at least on individual gene trees, the distances used
in STAR are essentially counts of clades, with gene samples being judged
closer when they occur together in more clades.

\medskip

Now fixing a species tree $\sigma$, the coalescent model determines
gene tree probabilities $\PP_\sigma(T)$. If $\mathcal C$ is a
non-trivial clade, and $\PP_\sigma(\mathcal C)$ denotes the
probability of clade $\mathcal C$, then
$$\prob_\sigma(\mathcal C)=
 \sum_{T \text{ displaying } \mathcal C} \prob_\sigma (T).$$

\begin{prop}\label{prop:cladeSTAR}
Suppose $\sigma = (\psi,  \lambda)$ is an $n$-taxon rooted
binary metric species tree.  Then the clade probabilities $\{\prob_\sigma(\mathcal C)\}$
determine $\EE_\sigma(D_T)$.
\end{prop}

\begin{proof}   Define two indicator functions:
$$
I_{\mathcal C} (T) = \begin{cases} 1 &\text{if $T$ displays $\mathcal C$,}\\ 0 &\text{otherwise,} \end{cases}
\hskip 1cm \text{and} \hskip 1cm
I_{a,b} (\mathcal C) = \begin{cases} 1 &\text{if $a,b\in\mathcal C$, }\\ 0 &\text{otherwise.} \end{cases}
$$

\medskip

Then using equation \eqref{eq:cd},
\begin{align}
\EE_\sigma(D_T(a,b)) &= \sum_{T} \prob_\sigma (T) \, 2(n - c_T(a,b))\notag\\
&=2n - 2 \sum_{T} \prob_\sigma (T) \, c_T(a,b)\notag\\
&=2n - 2 \sum_{T} \prob_\sigma (T) \bigg(\sum_{\text{non-trivial} \atop \text{clades } \mathcal C}  I_{\mathcal C}(T) I_{a,b} (\mathcal C) \bigg)\notag\\
&=2n - 2  \sum_{\text{non-trivial} \atop \text{clades } \mathcal C} \bigg( \sum_{T} \prob_\sigma (T) \, I_{\mathcal C}(T) \bigg) I_{a,b} (\mathcal C)\notag\\
&=2n - 2  \sum_{\text{non-trivial} \atop \text{clades } \mathcal C} \prob_\sigma (\mathcal C)  I_{a,b} (\mathcal C),\label{E:cprob}
\end{align}
and the expected STAR distance between $a$ and $b$ is computable from clade probabilities. 
\end{proof}

\begin{rmk}
  Though this statement and proof involve theoretical distributions
  and expected STAR distances, a similar statement and argument also
  apply to empirical clade frequencies and empirical STAR distances
  computed from a gene tree sample.
\end{rmk}

From Theorem \ref{thm:dist} we immediately obtain the following:

\begin{cor}
  The rooted species tree topoogy $\psi$ is identifiable from clade
  probabilities under the multispecies coalescent.
\end{cor}

\medskip

\begin{rmk} Proposition \ref{prop:cladeSTAR} also suggests an
  efficient algorithm for computing $\psi$ from $\{\prob ( \mathcal C
  )\}$, or an estimate of $\psi$ from estimates of $\{\prob ( \mathcal
  C )\}$.
\end{rmk}

We note that the program BUCKy \citep{bca2007} uses estimated clade
probabilities to construct a {\it concordance tree} from the clades
with highest estimated probabilities one clade at a time using a
greedy consensus approach.  Motivated by the observation that this
method can be misleading \citep{Degnan-etal-2009}, a quartet version
of BUCKy was developed that builds the species tree from the most
supported quartets \citep{Larget2010}.  However, the observation that
STAR distances can be computed from clade probabilities suggests an
alternative method for computing the concordance tree using Equation
\eqref{E:cprob} that would be consistent under the multispecies
coalescent.

\section{Discussion}

The main results of this paper are that the STAR algorithm and its
generalizations to other node numbering schemes, using any
well-behaved distance method for the selection of the species tree, do
provide statistically consistent methods of species tree inference. As
demonstrated by \citet{LiuSTAR}, the method is also fast and, at
least in simulations, exhibited good performance in comparison
to some other methods of species tree inference.

While our focus here has been theoretical, our work indicates several
insights that may be helpful in improving practical data analysis.
First, because the expected distance tables under STAR are
ultrametric, it may be preferable to use UPGMA rather than Neighbor
Joining as the tree building method in STAR, as UPGMA enforces an
ultrametric assumption while NJ does not. Alternatively, if NJ is
used, then even though the reconstructed distances on the inferred
species tree are ultimately discarded, it might be worthwhile to
first examine them to see if they roughly give ultrametricity. If not,
then one might doubt either the fit of the multispecies coalescent
model to the gene tree data, or that one had a sufficiently large
sample of gene trees for adequate inference. One should also consider
using a more elaborate distance method based on an optimality
criterion, such as Minimum Evolution, although this would require a
tree search and therefore should slow runtimes. Fast
heuristic searches under such a criterion, possibly adapted to enforce
ultrametricity, might improve STAR's performance. Simulation studies
are needed to further investigate these issues.

Second, the generalizations of the STAR method using alternate
numbering schemes also need further investigations. Although all such
schemes lead to consistent inference, this is an asymptotic statement
that concerns behavior with large samples. It is still possible that
some numbering schemes lead to more efficient inference, in that with
a small sample of gene trees they are more likely to return the
correct species tree, or one close to it. Understanding this issue,
perhaps for specific species trees, or better still for specific
models of species tree generation, would be desirable.

Even without more detailed theoretical results on the behavior of the
generalized STAR methods, they might still be useful in data analysis.
For instance, for a fixed sample of gene trees, one could repeatedly
infer a species tree using randomly chosen node numbering schemes
satisfying the criteria in Section \ref{sec:consistnode}. If all such
species trees have the same topology, one might be more confident in
the result. Significant differences in the inferred topologies might
lead one to doubt the validity of all the trees, again due to either
poor model fit of the coalescent or inadequate sample size. More
minor variations in the inferred trees might suggest which features of
the species tree one might be confident of, based on their shared
characteristics. Because of the computational speed of STAR,
calculations with many random node numberings should be
feasible.

\section{Acknowledgements}  
The work of E.~Allman and J.~Rhodes was
supported by the U.S. National Science Foundation (DMS
0714830), and that of J.~Degnan 
by the New Zealand Marsden Fund and the National Institute for
Mathematical and Biological Synthesis, an institute sponsored by the National Science Foundation, the
U.S. Department of Homeland Security, and the U.S. Department of Agriculture through NSF Award
\#EF-0832858, with additional support from the University of Tennessee, Knoxville.

\bibliographystyle{apalike}
\bibliography{STAR_node_numbering}

\end{document}